\documentclass[12pt,a4paper] {article}
    \usepackage{latexsym,amsmath,enumerate,amssymb,amsbsy,amsthm}
    \usepackage{enumerate,verbatim,tikz}
    
    \theoremstyle{plain}
    
    \newtheorem{theorem}{Theorem}
    \newtheorem{proposition}{Proposition}
    \newtheorem{corollary}{Corollary}
    \newtheorem{lemma}{Lemma}

    \theoremstyle{remark}
    \newtheorem{remark}[equation]{Remark}

    \numberwithin{equation}{section}

    \renewcommand\footnotemark{}
    
    \date{}
    
\begin{document}

        \title{SRIM and SCRIM Factors of $x^n+1$ over Finite Fields and Their Applications}

        \author{Arunwan Boripan and Somphong~Jitman}

        \thanks{A. Boripan is with the Department of Mathematics, Faculty of Science, Ramkhamhaeng University, Bangkok 10240, Thailand (email: boripan-arunwan@hotmail.com)}

        \thanks{S. Jitman (Corresponding Author)  is with the  Department of Mathematics, Faculty of Science,
            Silpakorn University, Nakhon Pathom 73000,  Thailand
            (email: sjitman@gmail.com).}

        \maketitle

    \begin{abstract} 
   Self-Reciprocal Irreducible Monic (SRIM) and  
  Self-Conjugate-Reciprocal Irreducible Monic (SRCIM)  factors of  $x^n-1$ over finite fields have become of interest due to their rich algebraic structures and wide applications.  In this paper, these notions are extended to factors of  $x^n+ 1$  over finite fields.   
   Characterization and enumeration of SRIM and SCRIM factors of $x^n+1$ over finite fields  are established.     Simplification and recessive  formulas for the number of such factors are given. Finally,   applications     in the studied of complementary negacyclic codes  are discussed.
    \end{abstract}
    
    \noindent{\bf keywords}: {negacyclic codes, complementary dual dual codes, reciprocal polynomials, SRIM polynomials, SCRIM polynomials}
    

\section{Introduction}
For a prime power $q$, let $\mathbb{F}_q$ denote  the finite field of order $q$. A non-zero  polynomial  $f(x)$   over a finite field $\mathbb{F}_q$  whose constant term is a unit in $\mathbb{F}_q$   is said to be  \emph{self-reciprocal}   if $f(x)$ equals its  {\em reciprocal polynomial} $f^*(x):=x^{\deg(f(x))}f(0)^{-1}f\left(\frac{1}{x}\right)$ and $f(x)$ is said to be  {\em self-reciprocal irreducible monic} (SRIM) if it is self-reciprocal, irreducible and monic.  For an element  $a$ in  $\mathbb{F}_{q^2}$, the {\em conjugate} of $a$ is defined to be $\overline{a}:=a^q$. The {\em conjugate} of a  polynomial $f(x)={f_0}+{f_1}x+\dots+{f_n}x^n$ in $\mathbb{F}_{q^2}[x]$ is defined to be  $\overline{f(x)}:=   \overline{f_0}+\overline{f_1}x+\dots+\overline{f_n}x^n$. 
A non-zero polynomial $f(x)$ over $\mathbb{F}_{q^2}$  whose constant term is a unit  is said to be {\em self-conjugate-reciprocal} if $f(x)$  equals its  {\em conjugate-reciprocal polynomial} $f^\dagger(x):= \overline{f^*(x)}$.  
	If, in addition, a self-conjugate-reciprocal $f(x)$  is monic and irreducible, it is said to be \emph{self-conjugate-reciprocal irreducible monic} (SCRIM).

     Due to their rich algebraic structures and wide applications,  SRIM and SCRIM polynomials  over finite fields have been extensively  studied and applied in various branches of Mathematics and Engineering.  SRIM polynomials have been used in  constructions of good  BCH codes in \cite{KG1969}.  In \cite{HB1975},   SRIM polynomials have been  characterized up to their degrees.    The orders and the number  of  SRIM polynomials of a given degree  over finite fields have been determined in  \cite{YM2004}.    Later, SRIM factors of $x^n-1$  have been used in the  characterization and enumeration of  Euclidean self-dual cyclic codes over finite fields in \cite{JLX2011}  and  \cite{PJ2018}.  They have been applied in the  characterization of   Euclidean complementary dual cyclic codes over finite fields in \cite{YM}. Recently, rigorous treatment for  SRIM factors of $x^n-1$ over finite fields has been given and applied in the study of  simple root Euclidean self-dual cyclic codes    over finite chain rings    in \cite{BSG2016}.     
    In \cite{BJU2015}, characterization of monic irreducible polynomials over  $\mathbb{F}_{q^2}$ to be SCRIM  are  given together with the  enumeration  of SCRIM polynomials of a fixed  degree.  Some properties of SCRIM factors of $x^n-1$ over  $\mathbb{F}_{ 2^{2l}}$  have been studied and applied in the characterization and enumeration of  Hermitian self-dual cyclic codes in \cite{JLS2014}.  In \cite{ASP2019}, a complete study of SCRIM factors of $x^n-1$ over $\mathbb{F}_{ q^{2}}$  has been given together with their applications. 
   While SRIM  and SCRIM factors of $x^n-1$ over finite fields are key to study Euclidean and Hermitian dual cyclic code of length $n$ over such fields, SRIM and SCRIM factors of $x^n+1$  will be useful in the study of   Euclidean and Hermitian dual of negacyclic codes of length~$n$ over finite fields.  
       
    In this paper, we focus on  SRIM and SCRIM factors of $x^n+1$  as well as their applications.
      Characterization and enumeration of SRIM   factors of $x^n+1$ over  finite fields are given in Section 3. Discussion on     SCRIM  factors of $x^n+1$  is  provided   in   Section 4.  In Section 5,  results on SRIM and SCRIM factors of $x^n+1$  over finite fields are applied in the studied of   complementary dual  negacyclic codes of length $n$ over  finite fields. Conclusion and remarks are given in Section 6.

\section{Preliminaries}
In this section, number theoretical results used   in this paper and preliminary results on SRIM and SCRIM factors of  $x^n\pm 1$ over finite fields   are recalled and discussed.

 \subsection{Number Theoretical Results}
 
 For integers $i\geq 0$ and $j\geq 1$,  denote by ${\rm o}^+_j(i)$ the additive order of $i$ modulo $j$. For coprime positive integers $i$ and $j$,  denote by  ${\rm ord}_j(i)$ the multiplicative order of $i$ modulo $j$.  For a prime $p$ and  integers $i\geq 0$ and  $j\geq 1$, we say that $p^i$ exactly divides $j$, denoted by $p^i || j$, if $p^i$ divides $j$ but $p^{i+1}$ does not divide $j$.

 For given nonzero coprime integers $a$ and $b$, a positive integer $d$ is said to be good  with respect to $a$ and $b$ if it is a divisor of $a^k+b^k$ for some positive integer $k$. In addition, if there exists an odd integer $r$ such that $d$ is a divisor of $a^r+b^r$,   $d$ is said to be oddly-good.  Denote by  $G_{(a,b)}$ and $OG_{(a,b)}$ the set of good integers and the set of oddly-good integers with respect to $a$ and $b$, respectively.  The following useful results on good and oddly-good in  \cite{J2017}, \cite{JPR2019}  and \cite{M1997} are recalled.

\begin{proposition}[{\cite[Theorem 1]{M1997}}] \label{Mor2}  Let $d>1$ be an odd integer. Then $ d \in G_{(a,b)}$ if and only if there exists an integer $s\geq 1$ such that $2^s|| {\rm ord}_p(\frac{a}{b})$ for every prime divisor  $p$  of  $d$. \end{proposition}
\begin{proposition}[{\cite[Proposition 3.2]{J2017}}] \label{oddly-good}  Let $d>1$ be an odd integer. Then $ d \in OG_{(a,b)}$ if and only if   $2|| {\rm ord}_p(\frac{a}{b})$ for every prime divisor $p$ of $d$. \end{proposition}
 
 \begin{proposition}[{\cite[Proposition 2.3]{JPR2019}}]  \label{prop2.01}Let $a$ and $b$   be  coprime odd  integers  and  let $\beta\geq 1$ be  an  integer. Then the following statements are equivalent.
     
     \begin{enumerate}
         \item  $2^\beta\in G_{(a,b)}$.
         \item  $2^\beta|(a+b)$.
         \item $2^\beta\in OG_{(a,b)}$.
     \end{enumerate} 
 \end{proposition}

 \begin{proposition}[{\cite[Proposition 2.2]{J2017}}] \label{Jit1}  Let $a,b,d>1$ be pairwise coprime odd integers. Then $ d \in G_{(a,b)}$ if and only if $ 2d \in G_{(a,b)}$. \end{proposition}
 
  \begin{proposition}[{\cite[Corollary 3.2]{J2017}}] \label{Jit2}  Let $a,b,d>1$ be pairwise coprime odd integers. Then $ d \in OG_{(a,b)}$ if and only if $ 2d \in OG_{(a,b)}$. \end{proposition}
  
 \begin{proposition}[{\cite[Proposition 2.7 and Corollary 3.2]{JPR2019}}]  \label{prop2m} Let $a,b$ and $d>1$   be pairwise coprime odd positive integers  and let  $\beta\geq 2$ be  an  integer. Then the following statements are equivalent. 
     \begin{enumerate}
         \item    $2^\beta d\in G_{(a,b)}$.
         \item   $2^\beta d\in OG_{(a,b)}$.
         \item   $2^\beta|(a+b)$
         and   $2|| {\rm ord}_{p}(\frac{a}{b})$ for every prime $p$ dividing $d$.
     \end{enumerate}
 \end{proposition}

 \subsection{Basic Properties of SRIM and SCRIM Factors of $x^n\pm 1$ over Finite Fields}
 For a positive integer $n$   and a  unit  $\lambda\in \mathbb{F}_q$, 
 denote by $\mathrm{SRIM}_{q,n}(\lambda)$ the  set of  SRIM factors of $x^n-\lambda$ in  $\mathbb{F}_q[x]$.   In the same fashion, denote by $\mathrm{SCRIM}_{q^2,n}(\lambda)$ the  set of  SCRIM factors of $x^n-\lambda$ in  $\mathbb{F}_{q^2}[x]$.

      We write  $n=p^\mu k$, where $p$ is the characteristic of $\mathbb{F}_q$, $p\nmid k$ and $\mu\geq 0$.  It follows that   $x^n\pm 1= (x^{k} \pm 1)^{p^\mu}$.  In the study of SRIM and SCRIM factors of $x^n\pm 1$, it is therefore  sufficient  to assume that    $\gcd(n,q)=1$.  Over a finite field of even characteristic, $x^n+1=x^n-1$ and its SRIM and SCRIM factors were completely studied in \cite{ASP2019} and  \cite{BSG2016}.  Without loss of generality,   $q$  is  assume to be an odd prime power throughout the study SRIM and SCRIM factors of $x^n+1$ over $\mathbb{F}_q$ and $\mathbb{F}_{q^2}$, respectively.  
      
      For each $0\leq i<n$,  the \emph{cyclotomic  coset of $q$ modulo  $n$  containing $i$} is defined to be the set 
      \begin{align*}
      Cl_{q,n}(i) =\{iq^{j} \, (\mathrm{mod}~ n) \mid j =0,1,2,\dots\}.
      \end{align*}
      It is not difficult to see that  $Cl_{q,n}(i)  =\{iq^{j} \, (\mathrm{mod } \,n) \mid  0\leq j < \mathrm{ord}_{\mathrm{o}_n^+(i)}{(q)}\}$ and $|Cl_{q,n}(i)  |=\mathrm{ord}_{\mathrm{o}_n^+(i)}{(q)}$. Moreover, $ {\mathrm{o}_n^+(i)}= {\mathrm{o}_n^+(j)}$ for all  $j\in Cl_{q,n}(i)$.
      Let $S_{q,n}$ denote a complete set of representatives of the cyclotomic  cosets of $q$ modulo  $n$ and let $\alpha$ be a primitive $n$th root of unity in some extension field of $\mathbb{F}_{q}$. It is well known (see \cite{LSBook}) that 
      \[f_i(x)=\prod_{j\in Cl_{q,n}(i) }{(x-\alpha^j)}\]
      is the minimal polynomial of $\alpha^i$  and it is a monic irreducible factor of $x^n-1$ over $\mathbb{F}_{q}$  for all     $0\leq i<n$.  Moreover,  we have 
      \begin{align}\label{xn-1}
      x^n-1=\prod_{i\in S_{q,n}} f_i(x).
      \end{align}

 For the case where $\lambda=1$,  the characterization and enumeration of    $\mathrm{SRIM}_{q,n}(1)$  and  $\mathrm{SCRIM}_{q^2,n}(1)$ have been given  in  \cite{ASP2019}, \cite{BSG2016} and \cite{SJLU2015}.
A general formula of the number of SRIM factors of $x^n-1$ over $\mathbb{F}_q$ can be determined in terms of good integers   as follows.

\begin{theorem}[{\cite{BSG2016}}] \label{SRIMxn-1}Let $q$ be a prime power and let $n$ be a positive integers coprime to $q$. Then the  number of SRIM factors of $x^n-1$ over $\mathbb{F}_{q}$ is
    \[\displaystyle  |\mathrm{SRIM}_{q,n}(1)|=\sum_{d\mid n,d\in G_{(q,1)}}\dfrac{\phi(d)}{\mathrm{ord}_{d}(q)},\]  
    where $\phi$ is the Euler's totient function.
\end{theorem}

Similarly,  the number of SCRIM factors of $x^n-1$ over $\mathbb{F}_{q^2}$ can be determined in terms of oddly-good integers   in the following theorem.
\begin{theorem} [{\cite{ASP2019}}]\label{SCRIMxn-1}Let $q$ be a prime power and let $n$ be a positive integers coprime to $q$. Then the number of SCRIM factors of $x^n-1$ over $\mathbb{F}_{q^2}$ is
    \[\displaystyle  |\mathrm{SCRIM}_{q^2,n}(1)|=\sum_{d\mid n, d\in OG_{(q,1)}}\dfrac{  \phi(d)}{\mathrm{ord}_{d}(q^2)},\] 
    where $\phi$ is the Euler's totient function.
\end{theorem}

Since $x^{2n}-1=(x^n-1)(x^n+1)$,  it can be deduced that  \begin{align} \label{eq-SRIM} \mathrm{SRIM}_{q,n}(-1)= \mathrm{SRIM}_{q,2n}(1) \setminus \mathrm{SRIM}_{q,n}(1)\end{align}
  and   \begin{align}\label{eq-SCRIM}
  \mathrm{SCRIM}_{q^2,n}(-1)= \mathrm{SCRIM}_{q^2,2n}(1) \setminus \mathrm{SCRIM}_{q^2,n}(1).\end{align}
  Hence, the next lemma follows. 

\begin{lemma} \label{lem2} Let $q$ be a prime power and let $n$ be a positive integer such that $\gcd(n,q)=1$.  Then the following statements hold. 
    \begin{enumerate}[$1)$]
        \item $|\mathrm{SRIM}_{q,n}(-1)|= |\mathrm{SRIM}_{q,2n}(1)| - |\mathrm{SRIM}_{q,n}(1)|$.
        
        \item    $|\mathrm{SCRIM}_{q^2,n}(-1)|= |\mathrm{SCRIM}_{q^2,2n}(1)|-| \mathrm{SCRIM}_{q^2,n}(1)|$.
    \end{enumerate}
    
\end{lemma}

  Form the discussion above,  characterization and enumeration of  SRIM and SCRIM factors of $x^n+1$ are given  in terms  of  SRIM and SCRIM factors of $x^{2n}-1$ and   $x^n-1$. 
  The goal of this paper is to give alternative    characterization and enumeration of  such factors of $x^n+1$.  Simplified characterization  and recursive formulas  for the enumeration of such polynomials are established. The results are presented    in Sections 3 and 4. Their applications in the study of negacyclic codes  are given in Section 5.

\section{SRIM Factors of $x^n+1$ over $\mathbb{F}_{q}$}

 In this section, the characterization and enumeration of SRIM factors of $x^n+1$ over $\mathbb{F}_q$  can be given   for all  positive integers $n$  such that $\gcd(n,q)=1$. 
 
 Here, we write $n=2^m n'$, where  $m\geq 0$ is an integer  and  $n'$ is an odd positive integer.

  \subsection{Characterization and Enumeration of SRIM Factors of $x^n+1$}

  For each  integer  $m\geq 0$   and  odd positive integer $n'$, let $\alpha$ be a primitive $2^{m+1}n'$th root of unity.  For each $0\leq i < 2^{m+1}n'$, let $f_i(x)= \prod\limits_{j\in Cl_{q, 2^{m+1}n'} (i)}(x-\alpha^j)$. 
 
  \begin{lemma} \label{srim-parity} Let $m\geq 0$ be an integer and let  $n'$ be an odd positive integer.   Let  $0\leq i < 2^{m+1}n'$. Then    elements in $Cl_{q, 2^{m+1}n'} (i)$  have the same parity.
 \end{lemma}
 \begin{proof}
  Let $a\in Cl_{q, 2^{m+1}n'} (i)$. Then $a\equiv iq^j ({\rm mod} \,2^{m+1}n')$ for some integer $j\geq 0$.   Then  $ a-i \equiv iq^j-i\equiv i(q^j-1)  ({\rm mod} \,2^{m+1}n')$.  Since $q^j-1$ is even,    we have   $ a-i \equiv 0 ({\rm mod} \,2)$.  Equivalently, $a$  and $i$ have the same parity.
 \end{proof}
 
 From Lemma \ref{srim-parity}, the parity of   a representative of $Cl_{q, 2^{m+1}n'} (i)$ is independent of its choices an hence we have the following lemma. 
 
  \begin{lemma}\label{div-1}  Let $m\geq 0$ be an integer and let  $n'$ be an odd positive integer.    Let $0\leq i < 2^{m+1}n'$. Then the following statements are equivalent.
      \begin{enumerate}
        \item $f_{i}(x) |(x^{2^mn'}+1)$.
       \item   $i$ is odd. 
       \item   $ 2^{m+1} |{\rm o}_{2^{m+1}n'}^+(i)$.
       \end{enumerate}
 \end{lemma}
 \begin{proof}
     To prove (1) implies (2), assume  that $f_{i}(x) |(x^{2^mn'}+1)$. Suppose that $i$ is even. Then $\frac{i}{2}$ is an integer and $0=(\alpha^i)^{2^mn'}+1=(\alpha^\frac{i}{2})^{2^{m+1}n'}+1= (\alpha^{2^{m+1}n'})^\frac{i}{2}+1=1+1=2$, a contradiction. Hence, $i$ is odd as desired.
     
     To prove (2) implies (3), assume  that $i$ is odd. Then ${\rm o}_{2^{m+1}n'}^+(i)=\frac{2^{m+1}n'} {\gcd(i, 2^{m+1}n')}= \frac{2^{m+1}n'} {\gcd(i, n')}$ which implies that $2^{m+1}||{\rm o}_{2^{m+1}n'}^+(i)$.
     
      To prove (3) implies (1), assume  that $2^{m+1}||{\rm o}_{2^{m+1}n'}^+(i)$.    Since   ${\rm o}_{2^{m+1}n'}^+(i)=\frac{2^{m+1}n'} {\gcd(i, 2^{m+1}n')} $,  $i$ is odd which implies that    $(\alpha^{i})^{2^mn'}  \ne 1$.  Since  $(\alpha^{i})^{2^{m+1}n'}  = (\alpha^{2^{m+1}n'})^{i} =  1$,  we have $(\alpha^{i})^{2^mn'}  +1=0$. It follows that 
     $\alpha^{it}+1 =0 $ for all  
      $t\in Cl_{q, 2^{m+1}n'} (i)$.  Hence, $f_{i}(x) |(x^{2^mn'}+1)$.
     \end{proof}
 
 Using the analysis similar to that of  \cite[Lemma 1]{JLX2011}, we have the following lemma. 
  \begin{lemma}  \label{charSRIM}
     Let $m\geq 0$ be an integer and let  $n'$ be an odd positive integer.    Let $0\leq i < 2^{m+1}n'$.  Then  the following statements are equivalent.
     \begin{enumerate}
         \item  $f_{i}(x)  $  is SRIM. 
         \item  $Cl_{q, 2^{m+1}n'} (i)=Cl_{q, 2^{m+1}n'} (-i)$.
         \item  $ {\rm o}_{2^{m+1}n'}^+(i)\in G_{(q,1)}$.
     \end{enumerate}
 \end{lemma}

Based on the above discussion, the following corollary can be deduced directly. 
  \begin{corollary} \label{charSRIMxn-1}
    Let $m\geq 0$ be an integer and let  $n'$ be an odd positive integer. Let  $0\leq i < 2^{m+1}n'$.  Then  $f_{i}(x)  $  is a SRIM  factor of   $x^{2^mn'}+1$  if and only if 
         $2^{m+1}|  {\rm o}_{2^{m+1}n'}^+(i)$  and $ {\rm o}_{2^{m+1}n'}^+(i)\in G_{(q,1)}$.
\end{corollary}
 
Consequently, the set of  SRIM factors of $x^n+1$ over $\mathbb{F}_{q}$  is 
 \begin{align}\label{uTd}
  \mathrm{SRIM}_{q,2^mn'}(-1)=\bigcup_{ d|n', 2^{m+1}d\in G_{(q,1)} }T_{d},
 \end{align}   
 where 
 \begin{align} \label{star1} T_d=\left\{ \prod_{j\in Cl_{q,2^{m+1}n'}(i)}{(x-\alpha^j)} \mid  0\leq i<2^{m+1}n' \text{ and  } {\rm o}_{2^{m+1}n'}^+(i)=2^{m+1} d\right\}\end{align}
 and 
 $\alpha$ is a primitive $2^{m+1}n'$th root of unity.

\begin{theorem} \label{enumSRIMxn+1} Let $m\geq 0$ be an integer and let  $n'$ be an odd positive integer. Let $\nu$ be a positive integer such that $ 2^\nu||(q+1)$. Then the  number of SRIM factors of $x^{2^m n'}+1$ over $\mathbb{F}_{q}$ is
    \[\displaystyle  |\mathrm{SRIM}_{q,2^mn'}(-1)|=\sum_{d\mid n',2^{m+1}d\in G_{(q,1)}}\dfrac{\phi(2^{m+1}d)}{\mathrm{ord}_{2^{m+1}d}(q)},\]  
    where $\phi$ is the Euler's totient function.
    
   If $m\geq \nu$, then $|\mathrm{SRIM}_{q,2^mn'}(-1)|=0$. In particular, if $q\equiv 3 ({\rm mod }\, 4)$,  then \[|\mathrm{SRIM}_{q,2^mn'}(-1)|=0\] for all $m\geq 1$.
\end{theorem}
\begin{proof} Let $d$ be a divisor of $n'$. Then the number of integers satisfied the conditions $0\leq i<2^{m+1}n'$  and  $ {\rm o}_{2^{m+1}n'}^+(i)=2^{m+1} d$ in \eqref{star1} is  ${\phi(2^{m+1}d)}$.  It is not difficult to see that each polynomial in $T_d$ has degree  $\mathrm{ord}_{2^{m+1}d}(q)$.  Hence, $|T_d|= \dfrac{\phi(2^{m+1}d)}{\mathrm{ord}_{2^{m+1}d}(q)}$.
    By \eqref{uTd}, the desired number $|\mathrm{SRIM}_{q,2^mn'}(-1)|$ is the summation of $|T_d|$ for all divisors $d$ of $n'$ such that  $2^{m+1}d\in G_{(q,1)}$.   
    
       Since $q$ is odd, we have that $\nu\geq 1$. Note that if $m\geq \nu$, we have   $2^{m+1}d \notin G_{(q,1)}$ for all divisors $d$ of $n'$ by Proposition \ref{prop2m}.  Hence, $|\mathrm{SRIM}_{q,2^mn'}(-1)|=0$.  Moreover, if $q\equiv 3 ({\rm mod }\, 4)$, we have $\nu=1$ which implies that      $2^{m+1}d \notin G_{(q,1)}$ and  $|\mathrm{SRIM}_{q,2^mn'}(-1)|=0$ for all $m\geq 1$.
    \end{proof}

  For an odd positive integer $n$, we observe that $x+1$ is  always  a SRIM factor of $x^n+1$.  
 In the following theorem, we focus on the characterization of  two extreme cases where  $x+1$ is  the only   SRIM factor of $x^n+1$ and where all the irreducible monic factors $x^n+1$  are SRIM.
  \begin{theorem} \label{charODDSRIMxn+1}
      Let $n$ be an odd   integer coprime to $q$. Then  the following statements hold.
      \begin{enumerate}
          \item  Every  monic irreducible factor of $x^{n}+1$ over $\mathbb{F}_{q}$ is SRIM if and only if 
          there exists a positive integer $s$ such that $2^s||\mathrm{ord}_l{(q)}$ for every   prime divisor $l$  of $n$. Equivalently, $n\in  G_{(q,1)}$.
          \item  The polynomial   $x+1$ is the only   SRIM  factor of $x^{n}+1$ over $\mathbb{F}_{q}$ if and only if 
          $\mathrm{ord}_l{(q)}$ is odd for all   prime divisors $l$  of $n$. 
      \end{enumerate}
  \end{theorem}
  \begin{proof} To prove (1),  assume that  every  monic irreducible factor of $x^{n}+1$ over $\mathbb{F}_{q}$ is SRIM.  From Theorem \ref{enumSRIMxn+1},  $2d\in G_{(q,1)}$ for  every divisor $d$ of $n$.  Hence,   $n\in G_{(q,1)}$ by   Proposition \ref{Jit1}.  The desired result is therefore follows from Proposition \ref{Mor2}.

      Conversely,   there exists a positive integer $s$ such that $2^s||\mathrm{ord}_l{(q)}$ for every   prime divisor $l$  of $n$.  By   Proposition \ref{Mor2}, we have that $n$ and all its divisors are in  $G_{(q,1)}$.  From Proposition \ref{Jit1}, $2d\in G_{(q,1)}$ for all divisors $d$ of $n$. Hence,  every  monic irreducible factor of $x^{n}+1$ over $\mathbb{F}_{q}$ is SRIM by  Theorem \ref{enumSRIMxn+1}.

      To prove (2),  assume that $x+1$ is the only   SRIM  factor of $x^{n}+1$ over $\mathbb{F}_{q}$. By Theorem \ref{enumSRIMxn+1},  $d=1$ is the only divisor of $n$ such that $2d\in G_{(q,1)}$.  Equivalently,   $d=1$ is the only divisor of $n$ such that $d\in G_{(q,1)}$ by Proposition \ref{Jit1}. Hence, $l\notin G_{(q,1)}$ for all prime divisors $l$ of $n$.   As desired, $\mathrm{ord}_l{(q)}$ is odd for all   prime divisors $l$  of $n$ by Proposition \ref{Mor2}.

      Conversely, $\mathrm{ord}_l{(q)}$ is odd for all   prime divisors $l$  of $n$. By Proposition \ref{Mor2},  $d=1$ is the only divisor of $n$ such that $ d\in G_{(q,1)}$.  Hence, it is the only divisor of $n$ such that $ 2d\in G_{(q,1)}$ by Proposition \ref{Jit1}. Hence, $|\mathrm{SRIM}_{q,n}(-1)|=1$ by Theorem \ref{enumSRIMxn+1}.

  \end{proof}

 \begin{proposition}  \label{eitherSRIMxn+1}
     Let $l$ be an odd prime  such that $l\nmid q$ and let $s$ be a positive integer. Then either  every  monic irreducible factor of $x^{l^s}+1$ over $\mathbb{F}_{q}$ is  SRIM or  $x+1$ is the only SRIM factor of $x^{l^s}+1$ over $\mathbb{F}_{q}$.
 \end{proposition}
 
 \begin{proof}  
    By Proposition \ref{Mor2} and  Proposition \ref{Jit1}, it follows that $ l\in G_{(q,1)}$ if and only if $ 2l^i\in G_{(q,1)}$ for all positive integers $i$. Hence,   every  monic irreducible factor of $x^{l^s}+1$ over $\mathbb{F}_{q}$ is  SRIM if and only if   $\mathrm{ord}_l{(q)}$ is even, and  $x+1$ is the only SRIM factor of $x^{l^s}+1$ over $\mathbb{F}_{q}$ if and only $\mathrm{ord}_l{(q)}$ is odd.
 \end{proof}

 The following theorems can be derived directly from  Theorem \ref{charODDSRIMxn+1}, Proposition~\ref{Mor2} and Proposition~\ref{Jit1}. The proofs are omitted.  
 
 \begin{theorem} \label{prod1SRIMxn+1}
     Let $n_1$ and $n_2$ be coprime odd integers relatively prime to $q$. If  $x+1$ is the only SRIM factor  of $x^{n_1}+1$ and $x^{n_2}+1$ over $\mathbb{F}_{q}$, then it is the only SRIM factor of $x^{n_1n_2}+1$ over $\mathbb{F}_{q}$.  
 \end{theorem}
 
 \begin{theorem} \label{prod2SRIMxn+1}
     Let $n_1$ and $n_2$ be coprime odd integers relatively prime to $q$. If every irreducible factor of $x^{n_1}+1$ over $\mathbb{F}_{q} $ is SRIM and  $x+1$ is the only SRIM factor of  $x^{n_2}+1$ over $\mathbb{F}_{q}$, then  $\mathrm{SRIM}_{q,n_1n_2}(-1)=\mathrm{SRIM}_{q,n_1}(-1)$. 
 \end{theorem}
 
 \begin{theorem} \label{prod3SRIMxn+1}
     Let $n_1$ and $n_2$ be  odd integers coprime  to $q$. If every   irreducible monic factor of $x^{n_1}+1$ and $x^{n_2}+1$ over $\mathbb{F}_{q}$ is SRIM, then all irreducible monic  factors of $x^{n_1n_2}+1$ over $\mathbb{F}_{q}$ are SRIM if and only if  there exists a positive integer $s$ such that $2^s|| \mathrm{ord}_{n_1}{(q)}$ and $2^s||\mathrm{ord}_{n_2}{(q)}$.
 \end{theorem}

\begin{corollary}\label{prod2prim-SRIMxn+1}
    Let $q$ be a prime power, let $l_1,l_2$ be distinct odd primes relatively to $q$ and let $r_1,r_2$ be  positive integers. Let $s_1$ and $s_2$ be non-negative integers such that $2^{s_1}|| \mathrm{ord}_{l_1}(q)$ and $2^{s_2}|| \mathrm{ord}_{l_2}(q)$. Then the following statements hold.
    \begin{enumerate}
        \item If $s_1=0=s_2$, then 
           \begin{align*}
        |\mathrm{SRIM}_{q,l_1^{r_1}l_2^{r_2}}(-1)|=1. 
        \end{align*}
        \item If  $s_1\geq 1$  and $s_2=0$, then
        \begin{align*}
        |\mathrm{SRIM}_{q,l_1^{r_1}l_2^{r_2}}(-1)|=|\mathrm{SRIM}_{q,l_1^{r_1}}(-1)|. 
        \end{align*}
        \item If $ s_1\ne s_2$ are positive, then
        \begin{align*}
        |\mathrm{SRIM}_{q,l_1^{r_1}l_2^{r_2}}(-1)|=|\mathrm{SRIM}_{q,l_1^{r_1}}(-1)|+|\mathrm{SRIM}_{q,l_2^{r_2}}(-1)|-1.
        \end{align*}
        \item If $ s_1=s_2\geq 1$, then
        \begin{align*}
        |\mathrm{SRIM}_{q,l_1^{r_1}l_2^{r_2}}(-1)|=\sum_{i=0}^{r_1}\sum_{j=0}^{r_2}\dfrac{\phi(l_1^il_2^j)}{\mathrm{ord}_{l_1^il_2^j}(q)}.
        \end{align*}
    \end{enumerate}
\end{corollary}
\begin{proof}   The first statement follows from Theorem  \ref{charODDSRIMxn+1} and Theorem \ref{prod1SRIMxn+1}.
   The second statement follows from   Theorem  \ref{charODDSRIMxn+1} and Theorem \ref{prod2SRIMxn+1}.
    
    Assume that $ s_1\ne s_2$ are positive. Based on Proposition \ref{Mor2}, Proposition \ref{Jit1} and Theorem  \ref{charODDSRIMxn+1},  we have   $\mathrm{SRIM}_{q,l_1^{r_1}l_2^{r_2}}(-1)=\mathrm{SRIM}_{q,l_1^{r_1}}(-1)\cup \mathrm{SRIM}_{q,l_2^{r_2}}(-1)$ and $ \mathrm{SRIM}_{q,l_1^{r_1}}(-1)\cap \mathrm{SRIM}_{q,l_2^{r_2}}(-1)=\{x+1\}$. Hence, 
   $ |\mathrm{SRIM}_{q,l_1^{r_1}l_2^{r_2}}(-1)|=|\mathrm{SRIM}_{q,l_1^{r_1}}(-1)|+|\mathrm{SRIM}_{q,l_2^{r_2}}(-1)|-1$.
    
    Assume that $ s_1=s_2\geq 1$.  From   Theorem  \ref{charODDSRIMxn+1} and Theorem \ref{prod3SRIMxn+1}, it follows that every monic irreducible factor of $x^n+1$ is SRIM. Hence, the enumeration formula in the fourth statement is obtained from Theorem \ref{enumSRIMxn+1}.
\end{proof}

 \subsection{Recursive Formulas for the Number  of SRIM Factors of $x^n+1$}
 
 In this subsection, we focus on recursive enumeration for SRIM factors of $x^n+1$ over $\mathbb{F}_{q}$. It can be given in terms of   the number   of SRIM factors of $x^{n'}-1 $   which is determined   \cite{BSG2016}, where $n'$ is the largest odd divisor of $n$.

\begin{theorem}[{\cite[Theorem 4.8]{BSG2016}}]
 \label{enumSRIMcyclic}
    Let $n'$ be an odd  positive integer relatively prime to $q$ and let $m$ be a non-negative integer.  Let $\nu$ be the positive integer such that $2^\nu||(q+1)$. Then  
    \begin{displaymath}
    |\mathrm{SRIM}_{q,2^mn'}(1)|=\left \{\begin{array}{ll}
    2|\mathrm{SRIM}_{q,n'}(1)|& \text{ if } m=1 \text{ or } m\geq 2 \text{ and }  \nu=1,\\ &\\
    2  |\mathrm{SRIM}_{q,n'}(1)| + (2^{\min\{m,\nu \}}-1)&(2|\mathrm{SRIM}_{q,n'}(1)|-|\mathrm{SRIM}_{q^2,n'}(1)| )\\
    & \text{ if } m\geq 2 \text{ and }  \nu\geq 2.
    \end{array}\right.
    \end{displaymath}
\end{theorem}

A recursive formula for the number of  SRIM factors of $x^n+1$ over $\mathbb{F}_{q}$ is given as follows.
\begin{theorem} \label{enumSRIM}
    Let $n'$ be an odd  positive integer relatively prime to $q$ and let $m$ be a non-negative integer. Let $\nu$ be the positive integer such that $2^\nu||(q+1)$. Then
    \begin{align*}
    |\mathrm{SRIM}_{q,2^{m}n'}(-1)|= \begin{cases}
     |\mathrm{SRIM}_{q,n'}(1)| & \text{if   } m=0\\ 
    3(2|\mathrm{SRIM}_{q,n'}(1)|-|\mathrm{SRIM}_{q^2,n'}(1)| )& \text{if }   m= 1 \text{ and } \nu\geq 2\\
       2^m  (2|\mathrm{SRIM}_{q,n'}(1)|-|\mathrm{SRIM}_{q^2,n'}(1)| )& \text{if }   2\leq m<\nu \text{ and } \nu\geq 2\\
     0& \text{if }   m\geq \nu \\
    \end{cases}
    \end{align*} 
 
\end{theorem}                                                                     \begin{proof}
From Lemma \ref{lem2}, we have
\begin{align} \label{eq-enumSRIM}
|\mathrm{SRIM}_{q,2^{m}n'}(-1)|=|\mathrm{SRIM}_{q,2^{m+1}n'}(1)|-|\mathrm{SRIM}_{q,2^{m}n'}(1)|.
\end{align}
Based on the  number $|\mathrm{SRIM}_{q,2^{m}n'}(1)|$ given Theorem \ref{enumSRIMcyclic},  it suffices to consider  the following $4$ cases.

\noindent {\bf Case 1:} $m=0$. From \eqref{eq-enumSRIM}, we have 
\begin{align*}
|\mathrm{SRIM}_{q,n'}(-1)|
&=|\mathrm{SRIM}_{q,2n'}(1)|-|\mathrm{SRIM}_{q,n'}(1)|\\
&=2|\mathrm{SRIM}_{q,n'}(1)|-|\mathrm{SRIM}_{q,n'}(1)|\\
&=|\mathrm{SRIM}_{q,n'}(1)|.
\end{align*}

\noindent {\bf Case 2:} $m=1$ and $\nu\geq 2$.
From \eqref{eq-enumSRIM}, we have 
\begin{align*}  
|\mathrm{SRIM}_{q,2n'}(-1)|
&=|\mathrm{SRIM}_{q,2^{2}n'}(1)|-|\mathrm{SRIM}_{q,2n'}(1)|\\
&=2  |\mathrm{SRIM}_{q,n'}(1)| + (2^{\min\{2,\nu \}}-1)(2|\mathrm{SRIM}_{q,n'}(1)|-|\mathrm{SRIM}_{q^2,n'}(1)| )\\
&~~~~-2|\mathrm{SRIM}_{q,n'}(1)|\\
&= 3(2|\mathrm{SRIM}_{q,n'}(1)|-|\mathrm{SRIM}_{q^2,n'}(1)| ).
\end{align*}

\noindent {\bf Case 3:} $2\leq m<\nu$ and $\nu\geq 2$.
From \eqref{eq-enumSRIM}, we have 
\begin{align*}  
|\mathrm{SRIM}_{q,2^{m}n'}(-1)|&=|\mathrm{SRIM}_{q,2^{m+1}n'}(1)|-|\mathrm{SRIM}_{q,2^{m}n'}(1)|\\
&=\left(  2  |\mathrm{SRIM}_{q,n'}(1)| + (2^{\min\{m+1,\nu \}}-1)(2|\mathrm{SRIM}_{q,n'}(1)|-|\mathrm{SRIM}_{q^2,n'}(1)| )\right)\\
&~~~-\left(  2  |\mathrm{SRIM}_{q,n'}(1)| + (2^{\min\{m,\nu \}}-1)(2|\mathrm{SRIM}_{q,n'}(1)|-|\mathrm{SRIM}_{q^2,n'}(1)| )\right)\\
&=(2^{\min\{m+1,\nu \}}-2^{\min\{m,\nu \}})(2|\mathrm{SRIM}_{q,n'}(1)|-|\mathrm{SRIM}_{q^2,n'}(1)| )\\
&=(2^{m+1}-2^{m})(2|\mathrm{SRIM}_{q,n'}(1)|-|\mathrm{SRIM}_{q^2,n'}(1)| )\\
&=2^{m}(2|\mathrm{SRIM}_{q,n'}(1)|-|\mathrm{SRIM}_{q^2,n'}(1)| ).
\end{align*}

\noindent {\bf Case 4:} $m\geq \nu$.
From \eqref{eq-enumSRIM}, we have 
\begin{align*}  
|\mathrm{SRIM}_{q,2^{m}n'}(-1)|
&=|\mathrm{SRIM}_{q,2^{m+1}n'}(1)|-|\mathrm{SRIM}_{q,2^{m}n'}(1)|\\
&=\begin{cases}
2|\mathrm{SRIM}_{q,n'}(1)|- 2|\mathrm{SRIM}_{q,n'}(1)| ~~~~~\text{ if } \nu=1,\\
&\\
\left(  2  |\mathrm{SRIM}_{q,n'}(1)| + (2^{\min\{m+1,\nu \}}-1)(2|\mathrm{SRIM}_{q,n'}(1)|-|\mathrm{SRIM}_{q^2,n'}(1)| )\right)\\
-\left(  2  |\mathrm{SRIM}_{q,n'}(1)| + (2^{\min\{m,\nu \}}-1)(2|\mathrm{SRIM}_{q,n'}(1)|-|\mathrm{SRIM}_{q^2,n'}(1)| )\right) \\
~~~~~~~~~~~~~~~~~~~~~~~~~~~~~~~~~~~~~~~~~~~~~~\text{ if } \nu\ge 2,
\end{cases}\\
&=\begin{cases}
 0 &\text{ if } \nu=1,\\
  (2^\nu-2^\nu)(2|\mathrm{SRIM}_{q,n'}(1)|-|\mathrm{SRIM}_{q^2,n'}(1)| )
&\text{ if } \nu\ge 2.
\end{cases}\\&=0.
\end{align*}
The proof is completed. 
\end{proof}

\section{SCRIM Factors of $x^n+1$ over $\mathbb{F}_{q^2}$}
 In this section, characterization  and enumeration of  SCRIM factors of $x^n+1$ over $\mathbb{F}_{q^2}$ are  given  for all odd prime powers $q$  and   positive integers $n$  such that $\mathrm{gcd}(n,q)=1$.

 \subsection{Characterization  of SCRIM Factors of $x^n+1$}

 For each  integer  $m\geq 0$   and  odd positive integer $n'$, let $\alpha$ be a primitive $2^{m+1}n'$th root of unity.  For each $0\leq i < 2^{m+1}n'$, let \[f_i(x)= \prod\limits_{j\in Cl_{q^2, 2^{m+1}n'} (i)}(x-\alpha^j).\]
 
 \begin{lemma}[{\cite[Lemma 3.2]{BJU2015} and  \cite[Lemma 3.5]{JLS2014}}] \label{charfSCRIM}
     Let $m\geq 0$ be an integer and let  $n'$ be an odd positive integer.    Let $0\leq i < 2^{m+1}n'$.  Then  the following statements are equivalent.
     \begin{enumerate}
         \item  $f_{i}(x)  $  is SCRIM. 
         \item  $Cl_{q^2, 2^{m+1}n'} (i)=Cl_{q^2, 2^{m+1}n'} (-qi)$.
         \item  $ {\rm o}_{2^{m+1}n'}^+(i)\in OG_{(q,1)}$.
     \end{enumerate}
 \end{lemma}
 
 From Lemma \ref{div-1} and Lemma \ref{charEiterSCRIM}, 
 \begin{corollary} \label{charfSCRIMxn+1}
     Let $m\geq 0$ be an integer and let  $n'$ be an odd positive integer. Let  $0\leq i < 2^{m+1}n'$.  Then  $f_{i}(x)  $  is a SCRIM  factor of   $(x^{2^mn'}+1)$  if and only if 
     $2^{m+1}|  {\rm o}_{2^{m+1}n'}^+(i)$  and $ {\rm o}_{2^{m+1}n'}^+(i)\in OG_{(q,1)}$.
 \end{corollary}

 Based on  the   discussion above, the set of  SCRIM factors of $x^n+1$ over $\mathbb{F}_{q^2}$  is 
 \begin{align}\label{uTdC}
 \mathrm{SCRIM}_{q^2,2^mn'}(-1)=\bigcup\limits_{ d|n', 2^{m+1}d\in OG_{(q,1)} } \mathcal{T}_d,
 \end{align}   
 where 
 \begin{align}\label{star2} 
     \mathcal{T}_d=\{ \prod\limits_{j\in Cl_{q^2,n}(i)}{(x-\alpha^j)} \mid  0\leq i<2^{m+1}n' \text{ and  } {\rm o}_{2^{m+1}n'}^+(i)=2^{m+1} d\}
 \end{align} and 
 $\alpha$ is a primitive $2^{m+1}n'$th root of unity. 
 
 The enumeration of SCRIM factors of $x^{2^m n'}+1$ over $\mathbb{F}_{q^2}$ is  given in the following theorem. 
 
 \begin{theorem} \label{enumSCRIMxn+1} Let $m\geq 0$ be an integer and let  $n'$ be an odd positive integer. Then the  number of SCRIM factors of $x^{2^m n'}+1$ over $\mathbb{F}_{q^2}$ is
     \[\displaystyle  |\mathrm{SCRIM}_{q^2,2^mn'}(-1)|=\sum_{d\mid n',2^{m+1}d\in OG_{(q,1)}}\dfrac{\phi(2^{m+1}d)}{\mathrm{ord}_{2^{m+1}d}(q^2)},\]  
 where $\phi$ is the Euler's totient function.
  
  If $m\geq \nu$, then $|\mathrm{SCRIM}_{q,2^mn'}(-1)|=0$. In particular, if $q\equiv 3 ({\rm mod }\, 4)$,  then $|\mathrm{SRIM}_{q,2^mn'}(-1)|=0$ for all $m\geq 1$.
 \end{theorem}
 \begin{proof}
    Let $d$ be a divisor of $n'$. Then the number of integers satisfied the conditions $0\leq i<2^{m+1}n'$  and  $ {\rm o}_{2^{m+1}n'}^+(i)=2^{m+1} d$ in \eqref{star2} is  ${\phi(2^{m+1}d)}$.  It is not difficult to see that each polynomial in $\mathcal{T}_d$ has degree  $\mathrm{ord}_{2^{m+1}d}(q^2)$.  Hence, $|\mathcal{T}_d|= \dfrac{\phi(2^{m+1}d)}{\mathrm{ord}_{2^{m+1}d}(q^2)}$.
    By \eqref{uTdC}, the desired number $|\mathrm{SCRIM}_{q^2,2^mn'}(-1)|$ is the summation of $|\mathcal{T}_d|$ for all divisors $d$ of $n'$ such that  $2^{m+1}d\in OG_{(q,1)}$.   
    
    Since $q$ is odd, we have that $\nu\geq 1$. Note that if $m\geq \nu$, we have  $2^{m+1}d \notin OG_{(q,1)}$ for all divisors $d$ of $n'$ by Proposition \ref{prop2m}.  Hence, $|\mathrm{SCRIM}_{q^2,2^mn'}(-1)|=0$.  Moreover, if $q\equiv 3 ({\rm mod }\, 4)$, we have $\nu=1$ which implies that    $2^{m+1}d \notin OG_{(q,1)}$  and  $|\mathrm{SCRIM}_{q^2,2^mn'}(-1)|=0$ for all $m\geq 1$.
 \end{proof}

 For an odd positive integer $n$, we observe that $x+1$ is always a SCRIM factor of $x^n+1$ over $\mathbb{F}_{q^2}$. 
Here, we focus on the characterization of  the two extreme cases  where $x+1$ is the only SCRIM factor of $x^n+1$ and where all the irreducible factors of $x^n+1$ are SCRIM.

 \begin{theorem} \label{charOddSCRIMxn+1}
     Let $n$ be an odd   integer coprime to $q$. Then the following statements hold.
     \begin{enumerate} \item  The following statements are equivalent. 
     \begin{enumerate}[$1)$]
         \item  Every  monic irreducible factor of $x^{n}+1$ over $\mathbb{F}_{q^2}$ is  SCRIM. 
         \item  For each  prime divisor $l$  of $n$,  every monic irreducible factor of $x^{l}+1$ over $\mathbb{F}_{q^2}$ is  SCRIM. 
         \item      $2|| \mathrm{ord}_l{(q)}$ for every prime divisor $l$ of $n$. Equivalently, $n  \in OG_{(q,1)}$. 
     \end{enumerate} 
 \item  The following statements are equivalent. 
 \begin{enumerate}[$1)$]
     \item  The polynomial $x+1$ is only the  SCRIM  factor of $x^{n}+1$ over $\mathbb{F}_{q^2}$. 
     \item  For each  prime divisor $l$  of $n$,  $x+1$ is only the  SCRIM  factor of $x^{l}+1$ over $\mathbb{F}_{q^2}$.
     \item    For each  prime divisor $l$ of $n$, $\mathrm{ord}_l{(q)}$ is odd or $4| \mathrm{ord}_l{(q)}$. Equivalently, $n  \notin OG_{(q,1)}$. 
 \end{enumerate} 
 \end{enumerate}
 \end{theorem}
 \begin{proof} 
    First, we prove (1).
    Since $(x^l+1)|(x^n+1)$ for all    prime divisors $l$  of $n$, the implication  1)$\Rightarrow$2)  follows.  
     
     Assume that   every monic irreducible factor of $x^{l}+1$ over $\mathbb{F}_{q^2}$ is  SCRIM for  all    prime divisors $l$  of $n$. Then $2l\in OG_{(q,1)}$ by Theorem \ref{enumSCRIMxn+1}. By Proposition \ref{Jit2}, we have $l\in OG_{(q,1)}$.  Hence, $2|| \mathrm{ord}_l{(q)}$  by Proposition \ref{oddly-good}. This proves the implication  2)$\Rightarrow$3).

     Assume that    $2|| \mathrm{ord}_l{(q)}$ for every prime divisor $l$ of $n$. Then by Proposition \ref{oddly-good},  $l\in OG_{(q,1)}$.  It follow hat $2l\in OG_{(q,1)}$ by Proposition \ref{Jit2}. From the formula in Theorem \ref{enumSCRIMxn+1},  it means that every  monic irreducible factor of $x^{n}+1$ over $\mathbb{F}_{q^2}$ is  SCRIM. The implication  1)$\Rightarrow$2)  is proved.

     Next,  we prove (2).     Since $(x^l+1)|(x^n+1)$ and $(x+1)|(x^l+1)$ for all prime divisors $l$  of $n$, the implication  1)$\Rightarrow$2)  follows.

   Assume that    $x+1$ is the only SCRIM factor of $x^{l}+1$ over $\mathbb{F}_{q^2}$   for  all    prime divisors $l$  of $n$. Then $2l\notin OG_{(q,1)}$ by Theorem \ref{enumSCRIMxn+1}. By Proposition \ref{Jit2}, we have $l\notin OG_{(q,1)}$.  Hence, $\mathrm{ord}_l{(q)}$ is odd or $4| \mathrm{ord}_l{(q)}$ by  Proposition \ref{oddly-good}.

  Assume that    $\mathrm{ord}_l{(q)}$ is odd or $4| \mathrm{ord}_l{(q)}$  for every prime divisor $l$ of $n$. Then by Proposition \ref{oddly-good},  $l\notin OG_{(q,1)}$.  It follow hat $2l\notin OG_{(q,1)}$ by Proposition \ref{Jit2}. From the formula in Theorem \ref{enumSCRIMxn+1},  it means that  $x+1$ is the only SCRIM  factor of $x^{n}+1$ over $\mathbb{F}_{q^2}$.
  
 \end{proof}

 \begin{proposition}  \label{charEiterSCRIM}
     Let $l$ be an odd prime  such that $l\nmid q$ and let $s$ be a positive integer. Then either  every  monic irreducible factor of $x^{l^s}+1$ over $\mathbb{F}_{q^2}$ is  SCRIM or  $x+1$ is the only SCRIM factor of $x^{l^s}+1$ over $\mathbb{F}_{q^2}$.
 \end{proposition}
 \begin{proof}
     By Proposition \ref{oddly-good} and  Proposition \ref{Jit2}, it follows that $ l\in OG_{(q,1)}$ if and only if $ 2l^i\in OG_{(q,1)}$ for all positive integers $i$. Hence,   every  monic irreducible factor of $x^{l^s}+1$ over $\mathbb{F}_{q^2}$ is  SCRIM if and only if   $2||\mathrm{ord}_l{(q)}$, and  $x+1$ is the only SCRIM factor of $x^{l^s}+1$ over $\mathbb{F}_{q^2}$ if and only $\mathrm{ord}_l{(q)}$ is odd or $4||\mathrm{ord}_l{(q)}$.
 \end{proof}

  The following theorems can be derived directly from  Theorem \ref{charOddSCRIMxn+1}, Proposition~\ref{oddly-good} and Proposition~\ref{Jit2}. The proofs will be omitted.

 \begin{theorem} \label{prod1SCRIM}
     Let $n_1$ and $n_2$ be coprime odd positive  integers relatively prime to $q$. If  $x+1$ is the only SCRIM factor  of $x^{n_1}+1$ and $x^{n_2}+1$ over $\mathbb{F}_{q^2}$, then it is the only SCRIM factor of $x^{n_1n_2}-1$ over $\mathbb{F}_{q^2}$.  
 \end{theorem}

 \begin{theorem} \label{prod2SCRIM}
     Let $n_1$ and $n_2$ be coprime odd positive integers relatively prime to $q$. If every irreducible factor of $x^{n_1}+1$ over $\mathbb{F}_{q^2} $ is SCRIM and  $x+1$ is the only SCRIM factor of  $x^{n_2}+1$ over $\mathbb{F}_{q^2}$, then  $\mathrm{SCRIM}_{q^2,n_1n_2}(-1)=\mathrm{SCRIM}_{q^2,n_1}(-1)$. 
    \end{theorem}

 \begin{theorem} \label{prod3SCRIM}
     Let $n_1$ and $n_2$ be  odd  positive integers coprime  to $q$. If every   irreducible monic factor of $x^{n_1}-1$ and $x^{n_2}-1$ over $\mathbb{F}_{q^2}$ is SCRIM, then all irreducible monic  factors of $x^{n_1n_2}-1$ over $\mathbb{F}_{q^2}$ are SCRIM. 
 \end{theorem}

 The next corollary follows directly from the discussion above and the arguments similar to those in the proof of Corollary \ref{prod2prim-SRIMxn+1}. 
 \begin{corollary} \label{enum2primeSCRIM}
     Let $q$ be a prime power, let $l_1,l_2$ be distinct odd primes relatively to $q$ and let $r_1,r_2$ be  positive integers. Let $s_1$ and $s_2$ be non-negative integers such that $2^{s_1}|| \mathrm{ord}_{l_1}(q)$ and $2^{s_2}|| \mathrm{ord}_{l_2}(q)$. Then the following statements hold.
     \begin{enumerate}
         \item If $s_1\ne 1$ and $s_2\ne1$, then 
         \begin{align*}
         |\mathrm{SCRIM}_{q^2,l_1^{r_1}l_2^{r_2}}(-1)|=1. 
         \end{align*}
         \item If  $s_1= 1$  and $s_2\ne 1$, then
         \begin{align*}
         |\mathrm{SCRIM}_{q^2,l_1^{r_1}l_2^{r_2}}(-1)|=|\mathrm{SCRIM}_{q^2,l_1^{r_1}}(-1)|. 
         \end{align*}
           
         \item If $ s_1=1=s_2$, then
         \begin{align*}
         |\mathrm{SCRIM}_{q^2,l_1^{r_1}l_2^{r_2}}(-1)|=\sum_{i=0}^{r_1}\sum_{j=0}^{r_2}\dfrac{\phi(l_1^il_2^j)}{\mathrm{ord}_{l_1^il_2^j}(q^2)}.
         \end{align*}
     \end{enumerate}
 \end{corollary}

\subsection{Recursive Formulas for the Number  of SCRIM Factors of $x^n+1$}

Note that the recursive enumeration for SCRIM factors of $x^n+1$ over $\mathbb{F}_{q^2}$ can be given independently from its   characterization.

For an odd positive integer $n'$, the number $|\mathrm{SCRIM}_{q^2,n'}(1)|$ of SCRIM factors of $x^{n'}-1 $    was given in \cite[Theorem 2.22]{ASP2019} as well as the calculation steps.  The  general formula for $  |\mathrm{SCRIM}_{q^2,2^mn'}(1)|$  was given in     \cite[Theorem 2.14]{ASP2019} for all  non-negative integers $m$. 

\begin{theorem}[{\cite[Theorem 2.14]{ASP2019}}] \label{evenInt}
    Let $n'$ be an odd  positive integer relatively prime to $q$ and let $m$ be a non-negative integer.  Let $\nu$ be the positive integer such that $2^\nu||(q+1)$. Then  
    \begin{displaymath}
    |\mathrm{SCRIM}_{q^2,2^mn}(1)|=\left \{\begin{array}{ll}
    2^m |\mathrm{SCRIM}_{q^2,n}(1)|& \text{ if } 0\leq m \leq \nu,\\
    2^\nu  |\mathrm{SCRIM}_{q^2,n}(1)|& \text{ if } m>\nu.
    \end{array}\right.
    \end{displaymath}
\end{theorem}

The   enumeration of SCRIM factors of $x^n+1$ over $\mathbb{F}_{q^2}$ is given as follows.
\begin{theorem} \label{enumSCRIM}
        Let $n$ be an odd  positive integer relatively prime to $q$ and let $m$ be a non-negative integer. Let $\nu$ be the positive integer such that $2^\nu||(q+1)$. Then
    \begin{align*}
    |\mathrm{SCRIM}_{q^2,2^{m}n}(-1)|= \begin{cases} 2^m|\mathrm{SCRIM}_{q^2,n}(1)| & \text{ if   }{0}\leq m< \nu,\\ 
    0 &\text{ if  } m\ge \nu.
    \end{cases}
    \end{align*} 
    In particular,     \begin{align*}
    |\mathrm{SCRIM}_{q^2,2^{m}}(-1)|= \begin{cases} 2^m& \text{ if   }{0}\leq m< \nu,\\
    0 &\text{ if  } m\ge \nu.
    \end{cases}
    \end{align*} 
\end{theorem}                                                                     
\begin{proof} From Lemma \ref{lem2}, we have
    \begin{align*}
    |\mathrm{SCRIM}_{q^2,2^{m}n}(-1)|=|\mathrm{SCRIM}_{q^2,2^{m+1}n}(1)|-|\mathrm{SCRIM}_{q^2,2^{m}n}(1)|.
    \end{align*}
    By Theorem \ref{evenInt}, it follows that 
    \begin{align*}
    |\mathrm{SCRIM}_{q^2,2^{m}n}(-1)|&=
    \begin{cases}
    2^{m+1}|\mathrm{SCRIM}_{q^2,n}(1)|-2^{m}|\mathrm{SCRIM}_{q^2,n}(1)| &\text{if } 0\leq m<\nu\\
      2^{\nu}|\mathrm{SCRIM}_{q^2,n}(1)|-2^{\nu}|\mathrm{SCRIM}_{q^2,n}(1)| &\text{if }  m\geq \nu
      \end{cases}\\
    &=\begin{cases} 2^m|\mathrm{SCRIM}_{q^2,n}(1)| & \text{ if   }{0}\leq m< \nu,\\ 
    0 &\text{ if  } m\ge \nu
    \end{cases}
    \end{align*}as desired.
\end{proof}

From Theorem \ref{evenInt} and Theorem \ref{enumSCRIM}, we have the following corollary. 
\begin{corollary}  
    Let $n$ be an odd  positive integer relatively prime to $q$ and let $m$ be a non-negative integer. If  $\nu$  is  the positive integer such that $2^\nu||(q+1)$, then
    \begin{align*}
        |\mathrm{SCRIM}_{q^2,2^{m}n}(-1)|=   |\mathrm{SCRIM}_{q^2,2^mn}(1)| & \text{ for all    }{0}\leq m< \nu.
    \end{align*} 
\end{corollary}

\section{Complementary Dual Negacyclic Codes over Finite Fields}
In this section, we focus on applications of SRIM and SCRIM factors of $x^n+1$ over finite fields in the study of complementary dual negacyclic codes.

\subsection{Polynomials and Codes}

A {\em linear   code} of length $n$   over $\mathbb{F}_q$ is defined to be a subspace of the $\mathbb{F}_q$-vector space $\mathbb{F}_q^n$.  The \textit{Euclidean dual} of a linear  code $C$ of length $n$ over $\mathbb{F}_q$ is defined to be  
\begin{align*}
C^{\perp_{\mathrm{E}}}=\{\boldsymbol{u}\in\mathbb{F}_q^n\mid \left\langle \boldsymbol{u},\boldsymbol{v}\right\rangle_{\mathrm{E}}=0 \mathrm{\,for\,all\,} \boldsymbol{v} \in C \},
\end{align*}
where 
$
\left\langle \boldsymbol{u},\boldsymbol{v}\right\rangle_{\mathrm{E}}:=\sum_{i=0}^{n-1}u_iv_i$
is the  \textit{Euclidean inner product} between  $\boldsymbol{u}=\left( u_0,u_1,...,u_{n-1}\right) $ and $\boldsymbol{v}=\left( v_0,v_1,...,v_{n-1}\right) $ in $\mathbb{F}_{q}^n$.
A linear  code $C$ is said to be \textit{Euclidean self-dual} if $C =C^{\perp_{\mathrm{E}}} $  and  $C$ is called \textit{Euclidean complementary dual}  if $C \cap C^{\perp_{\mathrm{E}}} =\{\boldsymbol{0}\}$.  Over $\mathbb{F}_{q^2}$,  the \textit{Hermitian dual} of a linear  code $C$ of length $n$  is defined to be  
\begin{align*}
C^{\perp_{\mathrm{H}}}=\{\boldsymbol{u}\in\mathbb{F}_{q^2}^n\mid \left\langle \boldsymbol{u},\boldsymbol{v}\right\rangle_{\mathrm{H}}=0 \mathrm{\,for\,all\,} \boldsymbol{v} \in C \}, 
\end{align*}
where $
\left\langle u,v\right\rangle_{\mathrm{H}}:=\sum_{i=0}^{n-1}u_i{v_i}^q$ is the  \textit{Hermitian inner product} between   $\boldsymbol{u}=\left( u_0,u_1,...,u_{n-1}\right) $ and $\boldsymbol{v}=\left( v_0,v_1,...,v_{n-1}\right) $ in $\mathbb{F}_{q^2}^n$.
A linear  code $C$ is said to be \textit{Hermitian self-dual} if $C =C^{\perp_{\mathrm{H}}} $  and $C$ is said to be  \textit{Hermitian complementary dual} if  $C \cap C^{\perp_{\mathrm{H}}} =\{\boldsymbol{0}\}$.

A linear code $C$ of length $n$  over $\mathbb{F}_q$ is called {\em negacyclic} if  $(-c_{n-1},c_0,...,c_{n-2})\in C$ for all  $(c_0,c_1,...,c_{n-1})\in C$.  As an extension on cyclic codes, negacyclic codes have been extensively studied due to rich algebraic structures and wide applications.  Self-dual negacyclic codes over finite fields have been quite well studied (see \cite{JPR2019}, \cite{SJLU2015} and references therein).  Here, we focus on complementary dual 
negacyclic codes over finite fields. Some properties of such codes have been discussed in  \cite{PZS2018}.  In the remaining parts of this section,  characterization and enumeration of complementary dual 
negacyclic codes over finite fields are given under both the Euclidean and Hermitian inner products in terms of SRIM and SCRIM factors of $x^n+1$ discussed in Sections 3--4. 

First, we recall  an  algebraic characterization of   negacyclic codes over finite fields.  A negacyclic code $C$ of length $n$ over $\mathbb{F}_q$ can be viewed as an embedding ideal in the principal ideal ring $\mathbb{F}_q[x]/ \left\langle x^n +1\right\rangle $, where $\left\langle x^n +1\right\rangle $ denotes the ideal generated by $x^n+1$ in $\mathbb{F}_q[x]$.  Moreover,  such an ideal is generated by a unique monic divisor of $x^n+1$ and it is called the {\em generator polynomial} of $C$. It is well know  (see \cite{SJLU2015}) that the Euclidean and Hermitian  duals of a negacyclic code are again negacyclic and their generator polynomials are given in the following theorems.

\begin{theorem} \label{gDE}
	Let $C$ be a negacyclic code of length $n$ over $\mathbb{F}_q.$ Then $C^{\perp_{\mathrm{E}}}$ is a negacyclic code of length $n$ over $\mathbb{F}_q$ with generator polynomial  $h^*(x)$, where $h(x)=\dfrac{x^n+1}{g(x)}.$ 
\end{theorem}

\begin{theorem}\label{gDH}
	Let $C$ be a negacyclic code of length $n$ over $\mathbb{F}_{q^2}.$ Then $C^{\perp_{\mathrm{H}}}$ is a negacyclic code of length $n$ over $\mathbb{F}_{q^2}$ with generator polynomial  $h^\dagger(x)$, where $h(x)=\dfrac{x^n+1}{g(x)}.$ 
\end{theorem}

The characterization and enumeration of complementary dual negacyclic codes are given in the following two subsections.

\subsection{Euclidean Complementary Dual Negacyclic Codes over Finite Fields}

Assume that the characteristic of $\mathbb{F}_q$ is $p$. For a positive integer $n$, write $n=p^\mu 2^mn'$, where $\mu$ and $m$ are non-negative integers and $n'$ is an odd positive integer such that $p\nmid n'$.  

Since the number of monic irreducible factors of  $x^{2^mn'}+1 $ over $\mathbb{F}_q$ is $r:=\sum\limits_{d\mid n' }\dfrac{\phi(2^{m+1}d)}{\mathrm{ord}_{2^{m+1}d}(q)}$ and $(f^*(x))^*=f(x)$ for all divisors $f(x)$ of $x^{2^mn'}+1 $,  the  polynomial $x^{n}+1 $ can be factorized  (see \cite[Equation (3.2)]{JPR2019}) of  the following form
\begin{align}\label{factSRIM}
x^n+1= (x^{2^mn'}+1)^{p^\mu}= \prod_{a=1}^s f_a(x)^{p^\mu}\times \prod_{b=1}^t h_b(x)^{p^\mu}h_b^*(x)^{p^\mu}  ,
\end{align}
where $f_a(x)$ is a SRIM factor of $x^{2^mn'}+1 $ for all $a=1,2,\dots, s:=|\mathrm{SRIM}_{q,2^m n'}(-1)|$,  and $h_b(x)$ is  an irreducible monic factor  of $x^{2^mn'}+1 $ which is not SRIM  for all $b=1,2,\dots, t:= \frac{r-s}{2}$.

\begin{theorem} \label{genE}
	Let $C$ be a negacyclic code of length  $n=p^\mu 2^mn'$ over $\mathbb{F}_q$ with generator polynomial $g(x)$.  Then $C$  is a Euclidean complementary dual  if and only if  $g(x)$ is self-reciprocal of the following form 
    \[g(x)=  \prod_{a=1}^sf_a(x)^{i_a}\times \prod_{b=1}^t(h_b(x)h_b^*(x))^{j_b}  ,\]
    where $i_a,j_b\in \{0,p^\mu\}$ for all  $a=1,2,\dots ,s$ and $b=1,2,\dots, t$.
	\end{theorem}
\begin{proof} We write  $g(x)=  \prod\limits_{a=1}^sf_a(x)^{i_a}\times \prod\limits _{b=1}^th_b(x) ^{j_b} h_b^*(x)^{k_b} ,$
    where   $i_a,j_b,k_b\in \{0,1,2,\dots, p^\mu\}$ for all  $a=1,2,\dots s$ and $b=1,2,\dots, t$. Then \[h(x)=  \prod\limits_{a=1}^sf_a(x)^{p^\mu-i_a}\times \prod\limits _{b=1}^th_b(x) ^{p^\mu-j_b} h_b^*(x)^{p^\mu- k_b} \]
     and 
    \begin{align} \label{gcd} \gcd(g(x),h^*(x)) =  \prod\limits_{a=1}^sf_a(x)^{\min\{i_a,p^\mu-i_a\}}\times \prod\limits _{b=1}^th_b(x) ^{\min\{j_b,p^\mu-k_b\}} h_b^*(x)^{\min\{k_b,p^\mu-j_b\}} .\end{align}
    
    Assume that  $C$  is a Euclidean complementary dual. Then  $C\oplus C^{\perp_{\rm E}}=\mathbb{F}_q^n$ is generated by   $\gcd(g(x),h^*(x)) =1$ by Theorem \ref{gDE}.  From \eqref{gcd}, it can be deduced that  $\min\{i_a,p^\mu-i_a\}=0$ for all $a=1,2,\dots,s$ and 
    ${\min\{j_b,p^\mu-k_b\}}=0={\min\{k_b,p^\mu-j_b\}}$ for all $b=1,2,\dots, t$.   Consequently,   $i_a,j_b=k_b\in \{0,p^\mu\}$ for all  $a=1,2,\dots s$ and $b=1,2,\dots, t$. The desired result follows.

    The converse is obvious. 
    \end{proof}

From Theorem \ref{genE},  	the number of Euclidean complementary dual  negacyclic codes of length $n=p^\mu 2^mn'$ over $\mathbb{F}_q$ is independent of $p^\nu$ and it  is determined in the following theorem. 

\begin{theorem} 
   The number of Euclidean complementary dual  negacyclic codes of length  $n=p^\mu 2^mn'$ over $\mathbb{F}_q$ is 
   \[2^{\frac{r+s}{2}}, \]
    where $r=\sum\limits_{d\mid n' }\dfrac{\phi(2^{m+1}d)}{\mathrm{ord}_{2^{m+1}d}(q)}$ and $s=|\mathrm{SRIM}_{q,2^m n'}(-1)|$.
\end{theorem}
\begin{proof} From Theorem \ref{genE},  it follow that  the  number of Euclidean complementary dual  negacyclic codes of length  $n=p^\mu 2^mn'$ over $\mathbb{F}_q$ is  $2^{s+t}=2^{s+\frac{r-s}{2}}= 2^{\frac{r+s}{2}}$.
    \end{proof}

\begin{remark} Let $\nu$ be a positive integer such that $2^\nu||(q+1)$.    From Theorem \ref{enumSRIM},  we have  $s=|\mathrm{SRIM}_{q,2^m n'}(-1)|=0$ for all $m\geq \nu$. Hence, 
 the  number of   Euclidean complementary dual  negacyclic codes of length  $p^\mu 2^mn'$ over $\mathbb{F}_q$ is  \[2^{\sum\limits_{d\mid n' }\dfrac{\phi(2^{m+1}d)}{2\mathrm{ord}_{2^{m+1}d}(q)}}\]
    for all $\mu\geq 0$ and $m\geq \nu$.

    \end{remark}

\subsection{Hermitian Complementary Dual Negacyclic Codes over Finite Fields}

Assume that the characteristic of $\mathbb{F}_{q^2}$ is $p$ and write $n=p^\mu 2^mn'$, where $\mu$ and $m$ are non-negative integers and $n'$ is an odd positive integer such that $p\nmid n'$.   Analogous results for Hermitian complementary dual negacyclic codes  are briefly  discussed. 

Since the number of monic irreducible factors of $x^{2^mn'}+1 $ over $\mathbb{F}_{q^2} $ is $r:=\sum\limits_{d\mid n' }\dfrac{\phi(2^{m+1}d)}{\mathrm{ord}_{2^{m+1}d}(q^2)}$ and $(f^\dagger(x))^\dagger=f(x)$ for all divisors $f(x)$ of $x^{2^mn'}+1 $,  the  polynomial $x^{n}+1 $ can  be factorized over $\mathbb{F}_{q^2}$ in the following form
\begin{align*} 
x^n+1= (x^{2^mn'}+1)^{p^\mu}= \prod_{a=1}^s f_a(x)^{p^\mu}\times \prod_{b=1}^t h_b(x)^{p^\mu}h_b^\dagger(x)^{p^\mu}  ,
\end{align*}
where $f_a(x)$ is a SCRIM factor of $x^{2^mn'}+1 $ for all $a=1,2,\dots, s:=|\mathrm{SCRIM}_{q^2,2^m n'}(-1)|$,  and $h_b(x)$ is a reducible monic  factor of $x^{2^mn'}+1 $ which is not  SCRIM for all $b=1,2,\dots, t:= \frac{r-s}{2}$.

Using Theorem \ref{gDH} and the augments similar to those in the proof of  Theorem~\ref{genE},  the characterization of a Hermitian complementary dual  negacyclic code follows. 

\begin{theorem} \label{genH}
    Let $C$ be a negacyclic code of length  $n=p^\mu 2^mn'$ over $\mathbb{F}_{q^2}$ with generator polynomial $g(x)$.  Then $C$  is a Hermitian complementary dual  if and only if  $g(x)$ is self-conjugate-reciprocal of the following form 
   \[g(x)=  \prod_{a=1}^sf_a(x)^{i_a}\times \prod_{b=1}^t(h_b(x)h_b^\dagger(x))^{j_b}  ,\]
 where $i_a,j_b\in \{0,p^\mu\}$ for all  $a=1,2,\dots, s$ and $b=1,2,\dots, t$.
\end{theorem}

From Theorem \ref{genH},  	the number of Hermitian complementary dual  negacyclic codes of length $n=p^\mu 2^mn'$ over $\mathbb{F}_{q^2}$ is independent of $p^\nu$ and it  is  summarized in the following theorem.

\begin{theorem} 
    The number of  Hermitian complementary dual   negacyclic codes of length  $n=p^\mu 2^mn'$ over $\mathbb{F}_{q^2}$ is 
    \[2^{\frac{r+s}{2}}, \]
    where $r:=\sum\limits_{d\mid n' }\dfrac{\phi(2^{m+1}d)}{\mathrm{ord}_{2^{m+1}d}(q)}$ and $s=|\mathrm{SCRIM}_{q,2^m n'}(-1)|$.
\end{theorem}

\begin{remark} Let $\nu$ be a positive integer such that $2^\nu||(q+1)$.    From Theorem \ref{enumSCRIM},  we have  $s=|\mathrm{SCRIM}_{q,2^m n'}(-1)|=0$ for all $m\geq \nu$. Hence, 
the  number of   Hermitian  complementary dual  negacyclic codes of length  $p^\mu 2^mn'$ over $\mathbb{F}_{q^2}$ is  \[2^{\sum\limits_{d\mid n' }\dfrac{\phi(2^{m+1}d)}{2\mathrm{ord}_{2^{m+1}d}(q^2)}}\]
for all $\mu\geq 0$ and $m\geq \nu$.

\end{remark}

\section{Conclusion and Remarks}
    SRIM and SCRIM  factors of  $x^n+ 1$  over finite fields have been studied.  Characterization and enumeration of  such polynomials have been established. Simplification and recessive  formulas for the number of such factors have been derived. Their applications      in the studied of complementary dual negacyclic codes  have been discussed.

   In general, it would be interesting  to study the set $\mathrm{SRIM}_{q,n}(\lambda)$   of  SRIM factors of $x^n-\lambda$ in  $\mathbb{F}_q[x]$ and the set $\mathrm{SCRIM}_{q^2,n}(\lambda)$  of  SCRIM factors of $x^n-\lambda$ in  $\mathbb{F}_{q^2}[x]$ for other units $\lambda$.

\section*{Acknowledgments}  This research was supported by the Thailand Research Fund and  Silpakorn University under Research Grant RSA6280042.

\end{document}